\documentclass[journal]{IEEEtran}										
\usepackage{rotating}
\usepackage{amsmath}
\ifCLASSINFOpdf
\usepackage{graphicx,amssymb}
\usepackage{algorithm}
 \usepackage{algorithmic}

\DeclareGraphicsExtensions{.pdf,.jpg,.jpeg,.png}
\else
\usepackage[dvips]{graphicx}
\usepackage{epsfig}
\DeclareGraphicsExtensions{.eps}
\fi

\usepackage{pdfsync}
\usepackage{amsthm}

\newtheorem{theorem}{Theorem}[section]

\newtheorem{corollary}[theorem]{Corollary}

\def\RR{ {\mathbb R} }
\def\ZZ{ {\mathbb Z} }
\numberwithin{equation}{section}

\begin{document}
\title{Reconstruction of sparse wavelet signals from partial  Fourier measurements}
\author{Yang Chen,~ Cheng Cheng ~and~ Qiyu Sun\thanks{Chen is with the Department of Mathematics,
Hunan Normal  University,
Changsha 100044, Hunan, China;
\and Cheng and Sun are with the Department of Mathematics,
University of Central Florida,
Orlando 32816, Florida, USA.
Emails: \ yang\_chenww123@163.com; cheng.cheng@knights.ucf.edu; qiyu.sun@ucf.edu. The project is partially supported by National Science Foundation (DMS-1412413).}}

\maketitle

\begin{abstract}
In this paper, we show that high-dimensional
sparse  wavelet signals of finite levels  can be constructed from
 their  partial Fourier measurements on a deterministic sampling set with cardinality
 about a multiple
   of signal sparsity.
\end{abstract}

\section{Introduction}

Sparse representation of signals in a 
dictionary
 has been used in signal processing, compression, noise reduction,
 source  separation, and many more  fields. 
Wavelet bases are well localized in time-frequency plane and they provide sparse representations of many signals and images
that have transient structures and singularities (\cite{daubechiesbook, mallatbook}).
In this paper, we consider recovering 
sparse wavelet signals of finite levels
from their partial Fourier measurements. 

  Let ${\bf D}$ be a dilation matrix with integer entries whose eigenvalues have modulus
 strictly larger than one, and  set $M=|\det {\bf D}|\geq 2$.
  Wavelet vectors ${\Psi}_m=(\psi_{m, 1}, \ldots, \psi_{m, r})^T, 1\le m\le M-1$,
  used in this paper are generated from a
 multiresolution analysis $\{V_j\}_{j\in \ZZ}$,
 a family of closed subspaces of $L^2:=L^2(\RR^n)$, that satisfies the following:
(i)\ $ V_j\subset V_{j+1}$ for all $j\in \ZZ$;
(ii)\ $V_{j+1}=\{f({\bf D}\cdot), \ f\in V_j\}$ for all $j\in \ZZ$;
(iii)\ $\overline{\cup_{j\in \ZZ}V_{j}}=L^{2}$;
(iv)\ $\cap_{j\in \ZZ}V_{j}=\{0\}$; and
(v)\  there exists a scaling vector $\Phi=(\phi_1,\cdots,\phi_r)^{T}\in V_{0}$ such that $\{\phi_l(\cdot -{\bf k}),\ 1\leq l\leq r, {\bf k}\in \ZZ^n\}$
is a Riesz basis for $V_0$ (\cite{daubechiesbook, mallatbook, grochenig92, cohen93, geronimo94, goodmanlee94, heller95, han98, bds99, sbhbook}).
They generate a Riesz basis
 $\{M^{j/2}{\Psi}_m({\bf D}^j{\bf x}-{\bf k}): 
 1 \le m \le M-1,\ {\bf k} \in \ZZ^n \}$ for the wavelet space $W_j:=V_{j+1}\ominus V_j$, the orthogonal complement of $V_j$ in  $V_{j+1}$, for every $j\in \ZZ$.
Therefore any signal $f$ in the scaling space $V_J$ of level $J\ge 0$ has a unique wavelet decomposition,
\begin{equation} \label{sparsesignal.def} 
f  =  f_0+ g_0+\cdots+ g_{J-1},
\end{equation}
where
\begin{equation} \label{sparsesignal.def1}
f_0=\sum_{k\in \ZZ^n} {\bf a}_0^T({\bf k}) {\pmb\Phi}(\cdot-{\bf k})\in V_0\end{equation}
 and
 \begin{equation}\label{sparsesignal.def2}
 g_j=\sum_{m=1}^{M-1}\sum_{{\bf k}\in \ZZ^n} {\bf b}_{m,j}^T({\bf k}) M^{j} {\pmb\Psi}_{m}({\bf D}^j\cdot-{\bf k})\in W_j, 0\le j\le J-1.\end{equation}
 In this paper,
 we consider  wavelet signals $f\in V_J$ with  $f_0$ and $g_j, 0\le j\le J-1$, in the above wavelet decomposition
having sparse representations.

   Define Fourier transform of an integrable function
$f$ on $\RR^n$ by
$$\hat f({\pmb \xi})=\int_{\RR^n} f({\bf t}) e^{-i{\bf t}\cdot {\pmb \xi}} d{\bf t}.$$
Due to  coherence of wavelet bases between different levels, the conventional optimization
method  does not work well  to   reconstruct a sparse wavelet signal $f$ of finite level
from its partial Fourier measurements
$\hat f(\xi), \xi\in \Omega$,
on a finite sampling set $\Omega$
 (\cite{donoho03, candes06, donoho06, baraniuk07, candes08, sun12, yang10b, foucartbook}).
 Recently, Prony's method was introduced in \cite{zhang, ccssampta15} for the exact reconstruction of one-dimensional sparse wavelet signals.

Denote by $\#E$ the cardinality of a  set $E$. We say that a wavelet signal $f\in V_J$ has sparsity ${\mathbf s}= (s_0, \cdots,  s_{J-1})$
if it has sparsity
$$s_j:= \left\{\begin{array} {l}
\max\{\# K_0, \# K_{1, 0}, \ldots, \# K_{M-1,0}\}  \ {\rm if} \ j=0\\
\max\{ \# K_{1, j},  \ldots, \# K_{M-1, j}\}\  {\rm if }\  j=1, \ldots, J-1,
\end{array}\right.$$
at level $j, 0\le j\le J-1$,
 where
 $K_0$ and $K_{m, j}$ are supports of coefficient vectors
$(a_0({\bf k}))_{{\bf k}\in \ZZ^n}$ and $(b_{m, j}({\bf k}))_{{\bf k}\in \ZZ^n}$
in the wavelet decomposition \eqref{sparsesignal.def},
\eqref{sparsesignal.def1}
and \eqref{sparsesignal.def2}
respectively. For the classcial one-dimensional scalar case (i.e. $n=1, r=1$ and ${\bf D}=2$),  under the assumption that Fourier transform of the  scaling function $\phi$
  does not vanish on $(-\pi, \pi)$,
  \begin{equation}\label{nonzerophi}
\hat\phi(\xi)\ne 0, \ \ \xi\in (-\pi, \pi),
\end{equation}
  Zhang and Dragotti proved in \cite{zhang} that a compactly supported  sparse wavelet signal
    of the form \eqref{sparsesignal.def}
  can be reconstructed from  its Fourier measurements
  on a  sampling set $\Omega$ of size about twice of its sparsity $s_0+\cdots+s_{J-1}$.
In  
this paper, we extend their result to high-dimensional
sparse wavelet signals
without nonvanishing condition \eqref{nonzerophi} on the scaling vector $\Phi$.
Particularly in Theorem \ref{maintheorem}, we show
  that
 any ${\mathbf s}$-sparse wavelet signal
 $f$   of the form \eqref{sparsesignal.def}
can be reconstructed from its Fourier measurements on a sampling set $\Omega$
with cardinality less than $2 Mr (s_0+\cdots+s_{J-1})$,
which is independent
on dimension $n$. 

\section{Multiresolution analysis and wavelets}
\label{multiwavelets.section}

Set ${\bf M}={\bf D}^T$. 
Then  the scaling vector ${\Phi}=(\phi_1, \ldots, \phi_r)^T$  of a multiresolution analysis
$\{V_j\}_{j\in \ZZ}$
  satisfies
 a matrix refinement equation,
\begin{equation}\label{refinefourier.def}
\widehat\Phi({\pmb\xi})={\bf G}_0({\bf M}^{-1}{\pmb \xi})\widehat\Phi({\bf M}^{-1}{\pmb \xi}),
\end{equation}
where the matrix function ${\bf G}_0$ of size $r\times r$ is bounded and  $2\pi$-periodic.
%
In this paper, we 
assume that  ${\bf G}_0$ has trigonometric polynomial entries. Hence $\Phi$ is compactly supported, 
and  the Riesz basis property for the scaling vector $\Phi$ can be reformulated as
that $(\widehat{\Phi}(\pmb \xi+2\pi{\bf k}))_{{\bf k}\in \ZZ^n}$
 has rank $r$ for every $\pmb\xi\in \RR^n$.
Therefore for any $\pmb\xi\in \RR^n$ there exist
${\bf k}(\pmb\xi, l)\in \ZZ^n, 1\le l\le r$, such that
\begin{equation} \label{rankcondition3.new}
\big(\widehat \Phi( \pmb\xi+2\pi{\bf  k})\big)_{{\bf k}\in \Lambda (\pmb\xi)}\
\text{has full rank} \ r,
\end{equation}
where
\begin{eqnarray}\label{Upsilon.def}
\Lambda (\pmb\xi) &\hskip-0.05in=\hskip-0.05in&\{
{\bf k}(\pmb\xi, l)\in \ZZ^n: \ 1\le l\le r\}.
 \end{eqnarray}

Let ${\bf p}_m,0\leq m \leq M-1$, be representatives of $\ZZ^n/{\bf M}\ZZ^n$, and write
 $$\ZZ^n=\bigcup_{m=0}^{M-1}({\bf p}_m+{\bf M}\ZZ^n).$$
Take matrices
 ${\bf G}_m$, $1\leq m\leq M-1$,    with trigonometric polynomial entries  such that
\begin{equation}\label{orthogonalcondition}
\sum_{m'=0}^{M-1}{\bf G}_0({\pmb\xi}+2\pi{\bf M}^{-1}{\bf p}_{m'})\overline{{\bf G}_{m}({\pmb \xi}+2\pi{\bf M}^{-1}{\bf p}_{m'})}^T=0
\end{equation}
for all $1\le m\le M-1$,
and
\begin{equation}\label{waveletrankcondition}
 {\bf G}(\pmb \xi) {\rm \ has\ rank} \ Mr {\rm\ for\ all}\ \pmb\xi\in \RR^n,
\end{equation}
where{\footnotesize
$${\bf G}(\pmb\xi)\hskip-0.02in = \hskip-0.05in \left(\hskip-0.05in \begin{array} {ccc}
{\bf G}_0(\pmb\xi+2\pi {\bf M}^{-1}{\bf p}_0)  & \hskip-0.05in \cdots & \hskip-0.05in {\bf G}_0(\pmb\xi+2\pi {\bf M}^{-1}{\bf p}_{M-1})\\
{\bf G}_1(\pmb\xi+2\pi{\bf M}^{-1}{\bf p}_0)  & \hskip-0.05in \cdots & \hskip-0.05in {\bf G}_1(\pmb\xi+2\pi{\bf M}^{-1}{\bf p}_{M-1})\\
\vdots   & \hskip-0.05in \ddots & \hskip-0.05in \vdots\\
{\bf G}_{M-1}(\pmb\xi+2\pi{\bf M}^{-1}{\bf p}_0)  &\hskip-0.05in \cdots & \hskip-0.05in {\bf G}_{M-1}(\pmb\xi+2\pi{\bf M}^{-1} {\bf p}_{M-1})\\
\end{array}\hskip-0.05in\right).$$}
In this paper,  wavelet vectors $\Psi_m, 1\leq m\leq M-1$, are defined as  follows:
\begin{equation}\label{wavelet.def}
\widehat\Psi_m({\pmb \xi})={\bf G}_m({\bf M}^{-1}{\pmb \xi})\widehat\Phi({\bf M}^{-1}{\pmb \xi}),\ \  1\leq m\leq M-1.
\end{equation}
Then  $\Psi_m$ are compactly supported 
and
 $\{M^{j/2}{\Psi}_m({\bf D}^j{\bf x}-{\bf k}): 
 1 \le m \le M-1,\ {\bf k} \in \ZZ^n \}$ forms a Riesz basis for the wavelet space $W_j:=V_{j+1}\ominus V_j$
 for every $j\in \ZZ$.

For the scaling vector $\Phi$ and wavelet vectors  $\Psi_m, 1\le m\le M-1$, constructed above, one may verify that any signal in $V_J$ has the unique wavelet decomposition
\eqref{sparsesignal.def}, \eqref{sparsesignal.def1} and \eqref{sparsesignal.def2}.

\section{Reconstruction of sparse wavelet signals}
\label{sparserecovery.section}

Take  ${\bf h}=(h_1, \ldots, h_n)\in \RR^n$ and sparsity vector ${\mathbf s}=(s_0, \ldots, s_{J-1})$,
and set $\|{\bf s}\|_\infty=\max_{0\le j\le J-1} s_j$.
 For $0\le j\le J-1$ and $0\le m\le M-1$,  let
\begin{eqnarray*}
\Gamma_{j}
& \hskip-0.1in = & \hskip-0.1in \{ (-s_j+1/2){\bf h}, (-s_j+3/2){\bf h}, \ldots, (s_j-1/2){\bf h}\},
\end{eqnarray*}
and
\begin{eqnarray*}
\Omega_{j} 
&\hskip-0.1in
= &\hskip-0.1in \cup_{\pmb\gamma\in\Gamma_{j}} \cup_{m=0}^{M-1}
 \big( \pi \pmb\gamma+2\pi{\bf M}^{j} {\bf p}_m\nonumber\\
& & \ + 2\pi  {\bf M}^{j+1}
\Lambda (\pi{\bf M}^{-j-1}\pmb\gamma+ 2\pi{\bf M}^{-1}{\bf p}_m)\big),
\end{eqnarray*}
where 
 the set $\Lambda$ of cardinality $r$ is defined by \eqref{Upsilon.def}.
Set
\begin{equation}\label{omega.def}
\Omega=\cup_{j=0}^{J-1}\Omega_{j}.
\end{equation}
Then 
$$\Omega\subset
\{ (-\|{\bf s}\|_{\infty}+1/2){\bf h}\pi, 
 \ldots, (\|{\bf s}\|_{\infty}-1/2){\bf h}\pi\}
+2\pi\ZZ^n,$$
and 
\begin{equation}
\# \Omega\le \sum_{j=0}^{J-1} \#\Omega_{j}=2 M r (s_0+s_1+\cdots+s_{J-1}). 
\end{equation}
The following is the main theorem of this paper.

\begin{theorem}\label{maintheorem}  Let ${\bf D}$ be a dilation matrix, 
$\Phi$ be a compactly supported scaling vector,
   $\Psi_m, 1\le m\le M-1$, be wavelet vectors satisfying \eqref{orthogonalcondition} and \eqref{waveletrankcondition},  let $\Omega$  be the set in \eqref{omega.def}
  with ${\bf h}=(h_1, \ldots, h_n)$.
 If   $1, h_1, \ldots, h_n$ are  linearly independent over the field of rationals, then any ${\mathbf s}$-sparse wavelet signal of the form
   \eqref{sparsesignal.def}, \eqref{sparsesignal.def1} and \eqref{sparsesignal.def2}
can be reconstructed from its  Fourier  measurements on $\Omega$.
\end{theorem}

\begin{proof} Let $f$ be an ${\mathbf s}$-sparse  signal with  wavelet representation  \eqref{sparsesignal.def},
\eqref{sparsesignal.def1} and \eqref{sparsesignal.def2}.
Set
\begin{equation}\label{maintheorem.pf.eq2}
\widehat{{\bf a}}_0(\pmb\xi)=\sum_{{\bf k}\in \ZZ^n} {\bf a}_0({\bf k}) e^{-i{\bf k}\cdot\pmb\xi},\end{equation}
and
\begin{equation}\label{bmj.eq}
\widehat{ {\bf b}}_{m,j}(\pmb\xi)=\sum_{{\bf k}\in \ZZ^n} {\bf b}_{m,j}({\bf k}) e^{-i{\bf k}\cdot\pmb\xi}\end{equation}
for $1\le m\le M-1$ and $0\le j\le J-1$. Then taking Fourier transform on both sides of the equation \eqref{sparsesignal.def}
gives
\begin{equation}\label{maintheorem.pf.eq1}
\hat f(\pmb\xi)= {\widehat{\bf a}}_0^{T}(\pmb\xi) \widehat {\Phi}(\pmb\xi)+\sum_{j=0}^{J-1} \sum_{m=1}^{M-1}\widehat{{\bf b}}_{m, j}^{T}({\bf M}^{-j}\pmb\xi) \widehat {\Psi}_m({\bf M}^{-j}\pmb\xi).
\end{equation}

\smallskip

Define  $f_i, 0\le i\le J-1$, by
\begin{equation} \label{maintheorem.pf.eq4}
\widehat f_i(\pmb\xi)={\widehat{\bf a}}^{T}_0(\pmb\xi) \widehat {\Phi}(\pmb\xi)+\sum_{j=0}^i \sum_{m=1}^{M-1}\widehat{{\bf b}}_{m, j}^{T}({\bf M}^{-j}\pmb\xi) \widehat {\Psi}_m({\bf M}^{-j}\pmb\xi).
\end{equation}
Then
\begin{equation}\label{level.eq0}
f_{J-1}=f,
\end{equation} and
\begin{eqnarray}\label{level.eq}
\widehat f_i({\bf M}^i\pmb\xi) &\hskip-0.05in = &\hskip-0.05in \widehat f_{i-1}({\bf M}^i{\pmb\xi})+
\sum_{m=1}^{M-1}{\widehat{{\bf b}}_{m, i}}^{T}(\pmb\xi) \widehat\Psi_m(\pmb\xi)\nonumber\\
 &\hskip-0.05in = &\hskip-0.05in {\widehat{{\bf a}}_i}^{T}(\pmb\xi) \widehat\Phi(\pmb\xi) +
\sum_{m=1}^{M-1}{\widehat{{\bf b}}_{m, i}}^{T}(\pmb\xi)\widehat\Psi_m(\pmb\xi)\nonumber\\
&\hskip-0.05in = &\hskip-0.05in \big({\widehat{{\bf a}}_i}^{T}(\pmb\xi){\bf G}_0({\bf M}^{-1}\pmb\xi)+\sum_{m=1}^{M-1}{\widehat{{\bf b}}_{m, i}}^{T}(\pmb\xi)\nonumber\\
&&\hskip0.1in \times {\bf G}_m({\bf M}^{-1}\pmb\xi)\big)\widehat\Phi({\bf M}^{-1}\pmb\xi)
\end{eqnarray}
for some vectors $\widehat{{\bf a}}_i(\pmb\xi)$ with trigonometric polynomial entries,
where the last equality  follows from \eqref{refinefourier.def} and \eqref{wavelet.def}.
\smallskip

For $0\le j\le J-1$, $\pmb\gamma\in \Gamma_{j}$ and $0\le m'\le M-1$, set
$$\eta_j(\pmb\gamma,m')=\pi{\bf M}^{-j}\pmb\gamma+2\pi{\bf p}_{m'}.$$  Applying \eqref{level.eq0} and \eqref{level.eq} with $i=J-1$, 
replacing $\pmb \xi$ in \eqref{level.eq} by $ \eta_{J-1}(\pmb\gamma,m')+2\pi{\bf M}{\bf k}, {\bf k}\in \Lambda({\bf M}^{-1} \eta_{J-1}(\pmb\gamma,m'))$, 
and using
 periodicity of $\widehat{{\bf a}}_{J-1}$ and $\widehat{{\bf b}}_{m,J-1}$,
 we obtain
\begin{eqnarray}\label{maintheorem.pf.eq60+}
\hskip-0.1in&\hskip-0.1in&\hskip-0.1in\hat f({\bf M}^{J-1}\eta_{J-1}(\pmb\gamma,m')+2\pi{\bf M}^{J}{\bf k})
\nonumber\\
\hskip-0.1in&\hskip-0.1in = &\hskip-0.1in A(J-1,\pmb\gamma,m')
  \widehat\Phi({\bf M}^{-1}\eta_{J-1}(\pmb\gamma,m')+2\pi{\bf k})
\end{eqnarray}
for all ${\bf k}\in \Lambda({\bf M}^{-1} \eta_{J-1}(\pmb\gamma,m'))$, 
where
\begin{eqnarray}\label{anonsingular}
\hskip-0.1in{}&\hskip-0.1in&\hskip-0.1in A(J-1,\pmb\gamma,m')={\widehat{\bf a}_{J-1}}^{T}(\pi{\bf M}^{-J+1}\pmb\gamma) {\bf G}_0({\bf M}^{-1}\eta_{J-1}(\pmb\gamma,m'))\nonumber\\
& \hskip-0.1in& \hskip-0.1in+\sum_{m=1}^{M-1} {\widehat{{\bf b}}_{m, J-1}}^{T}(\pi{\bf M}^{-J+1}\pmb\gamma){\bf G}_m({\bf M}^{-1}\eta_{J-1}(\pmb\gamma,m')).
\end{eqnarray}
 Recall from \eqref{rankcondition3.new} that
  \begin{equation}
  \Big(\widehat\Phi({\bf M}^{-1}\eta_{J-1}(\pmb\gamma,m')+2\pi{\bf k})\Big)_{{\bf k}\in \Lambda({\bf M}^{-1} \eta_{J-1}(\pmb\gamma,m'))}
  \end{equation}
  is nonsingular. Then $A(J-1,\pmb\gamma,m')$ can be solved from the linear system \eqref{maintheorem.pf.eq60+} for all $0\le m'\le M-1$ and $\pmb\gamma\in \Gamma_{J-1}$.

Recall from \eqref{orthogonalcondition} and \eqref{waveletrankcondition} that
{\small 
\begin{equation}\label{gnonsingular}
 {\bf G}(h\pi{\bf M}^{-J}\pmb\gamma)= \left(\hskip-0.05in \begin{array} {c}
{\bf G}_0({\bf M}^{-1}\eta_{J-1}(\pmb\gamma,m'))  \\
{\bf G}_1({\bf M}^{-1}\eta_{J-1}(\pmb\gamma,m'))  \\
\vdots   \\
{\bf G}_{M-1}({\bf M}^{-1}\eta_{J-1}(\pmb\gamma,m')) \\
\end{array}\hskip-0.05in\right)_{0\le m'\le M-1}
\end{equation}
}
is nonsingular.
Thus, for every $\pmb\gamma \in \Gamma_{J-1}$ and $1\le m\le M-1$,
\begin{equation}\label{abhat}
{\widehat{{\bf a}}_{J-1}}(\pi{\bf M}^{-J+1}\pmb\gamma)\ \ {\rm and} \ \
 {\widehat{{\bf b}}_{m, J-1}}(\pi {\bf M}^{-J+1} \pmb\gamma)
\end{equation}
 are uniquely determined 
from samples of $\hat f$ on $\Omega_{J-1} 
\subset \Omega$ by \eqref{anonsingular} and \eqref{gnonsingular}.

\smallskip
For $1\le m\le M-1$, it follows from
the linear independence assumption of $1, h_1, \ldots, h_n$  on the field of rationals that
$$e^{-i \pi {\bf k}\cdot {\bf M}^{-J+1} {\bf h}}, {\bf k}\in K_{m, J-1}, {\rm \ are\ distinct\ to\  each\ other}.$$
For $\pmb \gamma= n {\bf h}$ with $n\in \{-s_{J-1}+1/2, \ldots, s_{J-1}-1/2\}$,
\begin{equation}
\widehat{{\bf b}}_{m, J-1}(\pi {\bf M}^{-J+1} \pmb\gamma)=
\sum_{{\bf k}\in K_{m, J-1} } {\bf b}_{m,J-1}({\bf k}) \big(e^{-i \pi {\bf k}\cdot {\bf M}^{-J+1} {\bf h}}\big)^n
\end{equation}
by \eqref{bmj.eq}.
 Therefore applying  Prony's method (\cite{foucartbook,  zhang, kuma84, scharf91, barbieri92,  vmb02, osborne06, plonka13}) recovers trigonometric polynomials $\widehat{{\bf b}}_{m, J-1}, 1\le m\le M-1$,
from  their measurements on $\pi {\bf M}^{-J+1} \pmb\gamma, \pmb\gamma\in \Gamma_{J-1}$.
Hence 
${\bf b}_{m,J-1}({\bf k})$, ${\bf k}\in\ZZ^n$, can be recovered from samples of $\hat f$ on $\Omega$ for all $1\le m\le M-1$.

\smallskip

By the above argument,
 \begin{equation}\label{fJ-1}
f_{J-1}-f_{J-2}\ {\rm and}\ {\widehat f}_{J-2}(\pmb\xi),\  \pmb\xi \in\Omega,
\end{equation}
can be obtained from samples of $\hat f$ on $\Omega$, because
 \begin{eqnarray*}\label{le.eq}
\widehat f_{J-2}(\pmb\xi)  &\hskip-0.1in = &\hskip-0.1in    \hat f({\pmb\xi}) -
\sum_{m=1}^{M-1}\Big(\sum_{{\bf k}\in \ZZ^n} 
{\bf b}_{m,J-1}({\bf k}) e^{-i{\bf k}\cdot{\bf M}^{-J+1}\pmb\xi}\Big)\\
& & \qquad\qquad\quad \times \widehat\Psi_m({\bf M}^{-J+1}\pmb\xi)
\end{eqnarray*}
 by \eqref{maintheorem.pf.eq2} and \eqref{level.eq}.
Inductively we can reconstruct
\begin{equation}\label{finduct}
f_i-f_{i-1}\ {\rm and} \ {\widehat f}_{i-1}(\pmb\xi), \ \pmb\xi\in\Omega,
\end{equation}
from the samples of $\hat f$ on $\Omega$ for $ i=J-2, \cdots, 1$.

Taking $i=1$ in \eqref{finduct} determines  samples of ${\widehat f}_0$ on $\Omega$. 
Next we recover the function $f_0$ from its
Fourier measurements  on $\Omega_0\subset \Omega$.
By \eqref{maintheorem.pf.eq2} and \eqref{bmj.eq},
 $$\widehat{{\bf a}}_0(n\pi{\bf h})  =  \sum_{{\bf k}\in {\bf K}_0} {\bf a}_0({\bf k}) \big(e^{-i\pi {\bf k}\cdot {\bf h}}\big)^n
 $$
 and
 $$ \widehat{\bf b}_{m, 0}( n\pi{\bf h})\nonumber\\
  =\hskip-0.08in  \sum_{{\bf k}\in {\bf K}_{m,0}} {\bf b}_{m,0}({\bf k})\big(e^{-i\pi {\bf k}\cdot {\bf h}}\big)^n,
 $$ 
 where $n\in \{-s_0+1/2, -s_0+3/2, \ldots, s_0-1/2\}$. Similar to \eqref{abhat}, we can show that
\begin{equation*}
{\widehat{{\bf a}}_0}(n\pi {\bf h})\ \ {\rm and} \ \
 {\widehat{{\bf b}}_{m, 0}}(n\pi {\bf h}), \ 1\le m\le M-1,
\end{equation*}
 are uniquely determined 
from samples of $\widehat f_0$ on $\Omega$. Applying Prony's method again recovers ${\bf a}_0({\bf k})$ for ${\bf k}\in {\bf K}_0$ and ${\bf b}_{m, 0}({\bf k})$ for $1\le m\le M-1$ and ${\bf k} \in {\bf K}_{m,0}$. Therefore $f_0$ could be completely recovered from its Fourier measurements on $\Omega$. This together with \eqref{level.eq0}, \eqref{fJ-1} and \eqref{finduct} completes the proof.
\end{proof}

 The linear independence requirement on ${\bf h}=(h_1, \ldots, h_n)$ in Theorem \ref{maintheorem}
can be replaced by  a quantitative condition 
if
 the sparse  signal has
some additional  information  
 on its support, c.f. \cite{zhang}.

\begin{corollary}
Let ${\bf D}$, $\Phi$ and $\Psi_m, 1\le m\le M-1$, be as in Theorem \ref{maintheorem},
and let
$f$ be an ${\mathbf s}$-sparse signal in \eqref{sparsesignal.def} satisfying
\begin{equation}\label{support.additional}
 K_0\subset [a, b)^n \ \  {\rm and} \ \
K_{m,j}\subset {\bf D}^j[a, b)^n,
\end{equation}
where $1\le m\le M-1$, $0\le j\le J-1$ and $a<b$.
Then $f$ can be recovered from its Fourier measurements on $\Omega$ in \eqref{omega.def}
with ${\bf h}=(h_1, \ldots, h_n)$ satisfying
 \begin{equation}\label{hrequire}
 0<(b-a)(h_1+h_2+\cdots+h_n)\le 2.
 \end{equation}
\end{corollary}

\begin{proof} 
Following the argument in Theorem \ref{maintheorem}, it suffices to prove that
 $ e^{-i\pi {\bf k}\cdot {\bf h}}, {\bf k} \in  K_0$, are distinct,
and also that $ e^{-i\pi {\bf k} \cdot {\bf M}^{-j} {\bf h}}, {\bf k}\in  K_{m,j}$, are distinct
for every $1\le m\le M-1$ and $0\le j\le J-1$. The above distinctive property follows from \eqref{support.additional} and \eqref{hrequire} immediately.
\end{proof}

From the proof of Theorem \ref{maintheorem}, we have the following result on the  reconstruction of 
an $s$-sparse trigonometric polynomial from its  samples on a set of  size $2s$. 

\begin{corollary} Let ${\bf h}=(h_1, \ldots, h_n)$ with
$1, h_1, \ldots, h_n$ being linearly independent over the field of rationals, and
define
$$\Theta_s=
 \{ (-s+1/2){\bf h}, (-s+3/2){\bf h}, \ldots, (s-1/2){\bf h}\}, \ s\ge 1.
$$
Then any $n$-dimensional trigonometric polynomial
$$P(\pmb\xi)= \sum_{{\bf k}\in \ZZ^n} p({\bf k})e^{-i{\bf k}\cdot \pmb \xi}$$
with sparsity $s$,
$$\#\{{\bf k}:\ p({\bf k})\ne 0\}\le s,$$
can be reconstructed from its samples on $\Theta_s$.
\end{corollary}

\section{Simulations}
The following algorithm for sparse wavelet signal recovery is proposed in the proof of Theorem \ref{maintheorem}.

\noindent {\bf Algorithm 1:}
\begin{itemize}
\item[{1.}] Input sparsity vector ${\bf s}=(s_0,\cdots,s_{J-1})$.
\item[{2.}] Input Fourier measurements $\hat f(\pmb\xi)$, $\pmb\xi\in\Omega$ 
    and set $f_{J-1}=f$. 
\item[{3.}] {\bf for} $j=J-1$ to $0$ {\bf do}
\item[{}] \quad {\bf for} every $\pmb\gamma\in\Gamma_j$ {\bf do}
\item[{}] \qquad {\bf for} every $m'=0,\cdots,M-1$ {\bf do}
 \item[{}] \qquad \quad  3a)\ \ $\eta_j(\pmb\gamma,m')=\pi{\bf M}^{-j}\pmb\gamma+2\pi{\bf p}_{m'}$.
\item[{}] \qquad \quad 3b)\ \ Solve the linear system
\begin{eqnarray*}
\hskip-0.2in&\hskip-0.1in&\hskip-0.1in\big(\widehat f_j\big({\bf M}^{j}\eta_j(\pmb\gamma,m')+2\pi{\bf M}^{j+1}{\bf k}\big)\big)_{{\bf k}\in \Lambda (M^{-1}\eta_j(\pmb\gamma,m'))}\\
\hskip-0.2in&\hskip-0.1in=&\hskip-0.1in A(j,\pmb\gamma,m')\big(\widehat\Phi({\bf M}^{-1}\eta_j(\pmb\gamma,m')+2\pi{\bf k})\big)_{{\bf k}\in \Lambda (M^{-1}\eta_j(\pmb\gamma,m'))}
\end{eqnarray*}
\qquad \qquad \quad to get
\begin{eqnarray*}
&&\hskip-0.1in A(j,\pmb\gamma,m'):={\widehat{\bf a}_j}^{T}(\pi{\bf M}^{-j}\pmb\gamma) {\bf G}_0({\bf M}^{-1}\eta_j(\pmb\gamma,m'))\\
&\hskip0.05in&+\sum_{m=1}^{M-1} {\widehat{{\bf b}}_{m, j}}^{T}(\pi{\bf M}^{-j}\pmb\gamma){\bf G}_m({\bf M}^{-1}\eta_j(\pmb\gamma,m')).
\end{eqnarray*}
\item[{}] \qquad {\bf end for}
\item[{}] \qquad 3c)\ \ Solve the linear equation
{\small 
\begin{eqnarray*}
&&\big({\widehat{\bf a}_j}^{T}(\pi{\bf M}^{-j}\pmb\gamma), {\widehat{{\bf b}}_{1, j}}^{T}(\pi{\bf M}^{-j}\pmb\gamma),\cdots, {\widehat{{\bf b}}_{M-1, j}}^{T}(\pi{\bf M}^{-j}\pmb\gamma)\big)\\
&&\hskip0.1in\times{\bf G}(h\pi{\bf M}^{-j-1}\pmb\gamma)=\big(A(j,\pmb\gamma,0),\cdots,A(j,\pmb\gamma,M-1)\big).
\end{eqnarray*}
}
\item[{}] \quad {\bf end for}

\item[{}]\quad 3d)\ \  Recover ${\bf b}_{m, j}$ from $\widehat{{\bf b}}_{m, j}(\pi{\bf M}^{-j}\pmb\gamma)$, $\pmb\gamma\in\Gamma_j$ with Prony's method for every $1\le m\le M-1$.
\item[{}] \quad 3e)\ \ Subtract $\sum_{m=1}^{M-1}\widehat{{\bf b}}^{T}_{m, j}({\bf M}^{-i}\pmb \xi)\widehat{\Psi}_m({\bf M}^{-i}\pmb \xi)$ from $\hat f_j(\pmb\xi)$ to get $\hat f_{j-1}(\pmb\xi)$, $\pmb \xi\in\Omega$.
\item[{}] {\bf end for}
\item[{4.}] Recover ${\bf a}_0$ from $\widehat{{\bf a}}(\pi\pmb\gamma)$, $\pmb\gamma\in\Gamma_0$ with Prony's method.
\item[{5.}] Reconstruct the sparse wavelet signal
 \begin{eqnarray*}
 f({\bf t})&\hskip-0.1in=&\hskip-0.1in\sum_{{\bf k}\in\ZZ^n}{\bf a}_0^T({\bf k})\Phi({\bf t}-{\bf k})\\
 &&+\sum_{j=0}^{J-1}\sum_{m=1}^{M-1}\sum_{{\bf k}\in\ZZ^n}{\bf b}_{m,j}^T({\bf k})M^j\Psi_m({\bf D}^{j}{\bf t}-{\bf k}).
 \end{eqnarray*}
\end{itemize}

Next we present simulations to demonstrate the above algorithm for  perfect reconstruction of sparse wavelet signals of finite levels.
Let
$\phi_1(t)=\chi_{[0,1)}(t)$ and $\phi_2(t)=2\sqrt{3}(t-1/2)\chi_{[0,1)}(t)$ be scaling functions, and let
$$\psi_1(t)=(6t-1)\chi_{[0,1/2)}(t)+(6t-5)\chi_{[1/2,1)}(t),$$
and
$$\psi_2(t)= 2\sqrt{3}(2t-1/2)\chi_{[0,1/2)}(t)-2\sqrt{3}(2t-3/2)\chi_{[1/2,1)}(t)$$
be wavelet functions.
Consider reconstructing the  sparse signal 
\begin{eqnarray}\label{signaldefi}
f(t)&\hskip-0.1in=&\hskip-0.1in {\bf a}_{0}^{T}(2)\Phi(t-2)+{\bf a}_0^{T}(4)\Phi(t-4)\nonumber\\
&&\hskip-0.1in+{\bf b}^{T}_{0}(1)\Psi(t-1)+{\bf b}^{T}_{0}(5)\Psi(t-5)\nonumber\\
&&\hskip-0.1in+{\bf b}^{T}_{1}(6)\Psi(2t-6)+{\bf b}^{T}_{1}(12)\Psi(2t-12)
\end{eqnarray}
from its Fourier measurements on the sampling set
\begin{equation*}
\Omega=\Big\{- \frac{\sqrt{2}}{128} n \pi+2k\pi:\ n=\pm 1, \pm 3 \ {\rm and}  \ k=0,\pm1,\pm2,4\Big\}
\end{equation*}
in \eqref{omega.def}, 
where $\Phi=(\phi_1,\phi_2)^{T}$,
$\Psi=(\psi_1, \psi_2)^T$, and the nonzero components of ${\bf a}_{0}$, ${\bf b}_{0}$ and ${\bf b}_{1}$ are randomly chosen in $[-1,1]\setminus(-0.1,0.1)$, see Figure \ref{signal.fig}.
Applying the proposed algorithm, 
our numerical results support the conclusion on perfect recovery of sparse wavelet signals 
 from their Fourier measurements on $\Omega$.
 \begin{figure}[h]
 \begin{center}
\includegraphics[width=42mm, height=30mm]{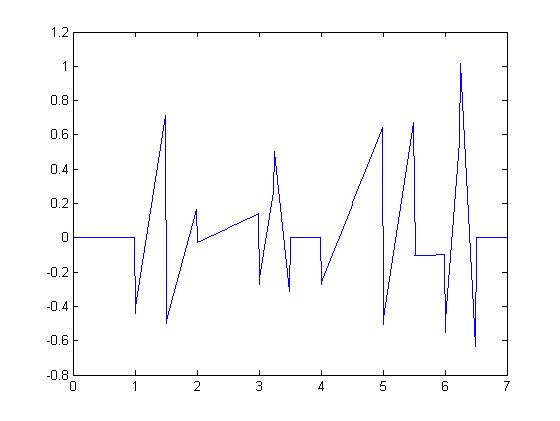}
\includegraphics[width=42mm, height=30mm]{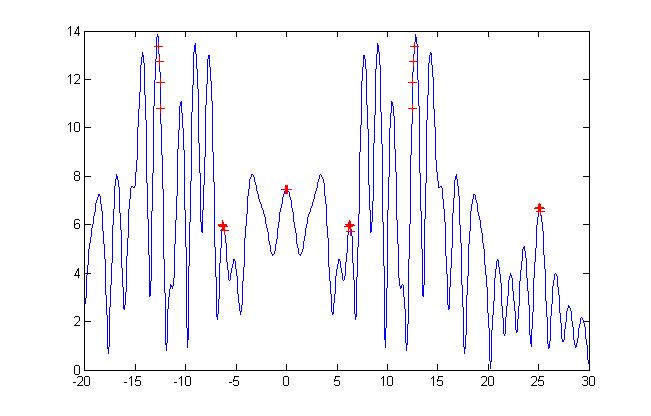}
\caption{Plotted on the left is the sparse wavelet signal $f$ in \eqref{signaldefi}, while on the right is the magnitude of its Fourier transform and the measurements on $\Omega$.}
\label{signal.fig}
\end{center}
\end{figure}

The proposed algorithm is tested when the Fourier measurements of the signal $f$ are corrupted by random noises $\epsilon$,
$$h(\xi)=\hat f(\xi)+\epsilon(\xi),\ \ \xi\in\Omega.$$
In this case, sparsity locations obtained by Prony's method in the algorithm are not necessarily integers, but it is observed that they are not far away from the sparsity locations of the signal $f$, when the signal-to-noise-ratio (SNR),
$$SNR=-20\log_{10}{\frac{\max_{\xi\in\Omega}|\epsilon(\xi)|}{\max_{\xi\in\Omega}|\hat f(\xi)|}}$$
is above 50 dB. Taking nearest integers of those locations may perfectly recover the sparsity positions $\{2,4\}$ for the scaling component of level $0$, $\{1,5\}$ for the wavelet component of level $0$, and $\{6,12\}$ for the wavelet component of level $1$. Then the signal $f$ can be reconstructed by the proposed algorithm approximately, see Figure \ref{error.fig}. 
 \begin{figure}[h]
 \begin{center}
\includegraphics[width=68mm, height=30mm]{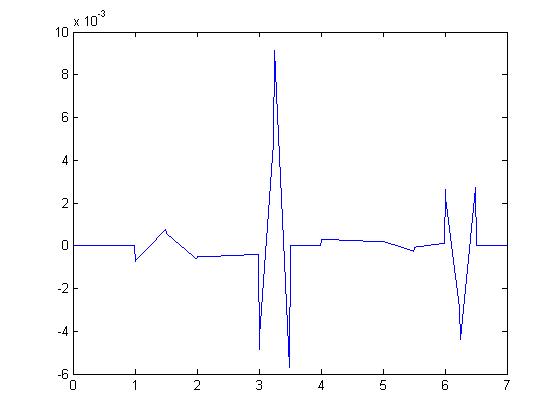}
\caption{The difference between the original signal and the reconstructed signal. In this figure, the reconstruction is generated by the proposed algorithm with modified Prony's method, and the noise level on Fourier measurements is SNR=50 dB. }
\label{error.fig}
\end{center}
\end{figure}

We also tested our proposed algorithm for two-dimensional wavelet signals with dilation
${\bf D}=\left(\hskip-0.05in \begin{array}{cc}
0 & -2\\
1 & 0\\
\end{array}\hskip-0.05in\right)$. Presented on the left of Figure \ref{2dsignalfourier.fig} is the amplitude of a sparse wavelet signal
\begin{eqnarray}\label{twodimensionalsignaldefi}
\hskip-0.3in f(t_1,t_2)&\hskip-0.1in=&\hskip-0.1in  a_0\phi(t_1-1,t_2)+ a_1\phi(t_1-2,t_2-3)\nonumber\\
&\hskip-0.1in+&\hskip-0.1in b_{0}\psi(t_1-2,t_2-1)+ b_{1}\psi(t_1-3,t_2-5),
\end{eqnarray}
where $a_0, a_1, b_0, b_1\in [-1,1]\setminus(-0.1,0.1)$ are selected randomly, the scaling function is
$\phi(t_1,t_2)=\chi_{[0,1)}(t_1)\chi_{[0,1)}(t_2)$, and the wavelet function is
$\psi(t_1,t_2)=\chi_{[0,1)}(t_1)(\chi_{[0,1/2)}(t_1)-\chi_{[1/2,1)}(t_2))$.
Our simulations show that the signal $f$ in \eqref{twodimensionalsignaldefi} can be reconstructed from its Fourier measurements on
\begin{equation}\label{omegatwodimensional}
\Omega=\Big\{\Big(\frac{\sqrt{2}}{64} n+2k, \frac{\sqrt{3}}{64} n+2l\Big)\pi,\ n=\pm 1, \pm 3 \ {\rm and}\   k,l=0, 1\Big\},\end{equation}
which is plotted on the right of Figure \ref{2dsignalfourier.fig}.
\begin{figure}[h]
 \begin{center}
\includegraphics[width=42mm, height=30mm]{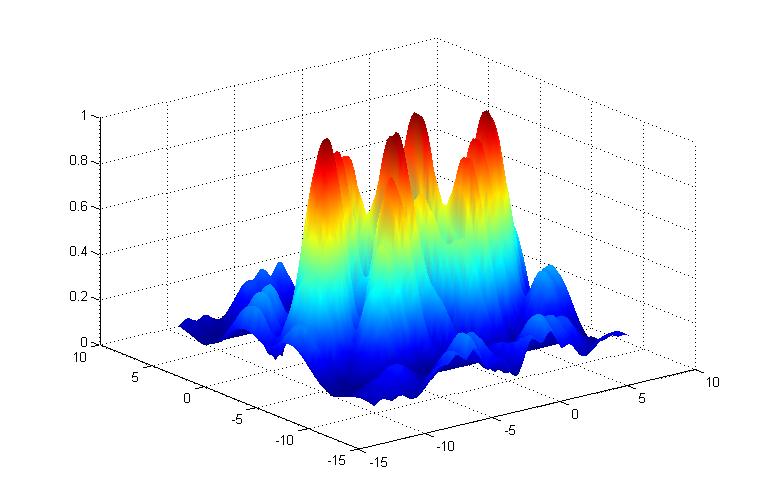}
\includegraphics[width=42mm, height=30mm]{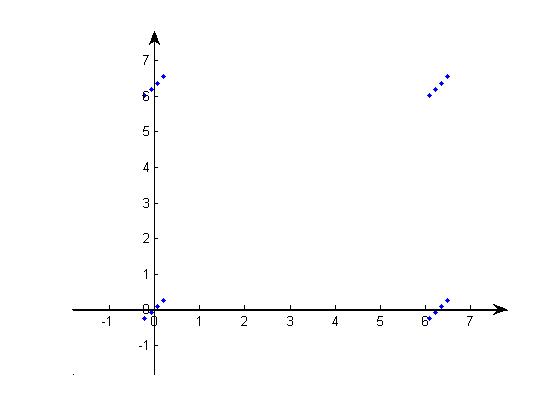}
\caption{Fourier amplitudes of the signal $f$ in \eqref{twodimensionalsignaldefi} and the sampling set $\Omega$
 in \eqref{omegatwodimensional}
 for sparse recovery. }
\label{2dsignalfourier.fig}
\end{center}
\end{figure}

\section{Conclusion} 

In this paper, we show that  sparse  wavelet signals of finite level
can be reconstructed from  their Fourier measurements on a deterministic sampling set, whose
cardinality is  independent on signal dimension and almost proportional to signal sparsity. A difficult problem on this aspect is
exact reconstruction of signals having sparse wavelet-like (e.g. wavelet packet, framelet, curvelet,  and shearlet) representations
from their partial Fourier information (\cite{coifman92, ronshen97,  candes06siam,  chui06, kutyniokbook}).

\end{document}